\documentclass[oneside,english]{amsart}
\usepackage[T1]{fontenc}
\usepackage[latin9]{inputenc}
\usepackage{float}
\usepackage{url}
\usepackage{amstext}
\usepackage{amsthm}
\usepackage{amssymb}
\usepackage{graphicx}
\PassOptionsToPackage{normalem}{ulem}
\usepackage{ulem}

\makeatletter

\newcommand*\LyXZeroWidthSpace{\hspace{0pt}}
\floatstyle{ruled}
\newfloat{algorithm}{tbp}{loa}
\providecommand{\algorithmname}{Algorithm}
\floatname{algorithm}{\protect\algorithmname}

\numberwithin{equation}{section}
\numberwithin{figure}{section}
\newenvironment{lyxcode}
	{\par\begin{list}{}{
		\setlength{\rightmargin}{\leftmargin}
		\setlength{\listparindent}{0pt}
		\raggedright
		\setlength{\itemsep}{0pt}
		\setlength{\parsep}{0pt}
		\normalfont\ttfamily}%
	 \item[]}
	{\end{list}}
\theoremstyle{plain}
\newtheorem{thm}{\protect\theoremname}[section]

\usepackage{a4wide}

\usepackage{bbm}
\usepackage{enumerate}
\usepackage{tikz}

\allowdisplaybreaks

\newcommand{\R}{\mathbb{R}}

\newcommand{\E}{\mathbb{E}}

\newcommand{\N}{\mathbb{N}}

\colorlet{mycol}{black}

\newtheorem{hypothesis}{Assumption}

\date{\today}

\AtBeginDocument{
  
}

\makeatother

\usepackage{babel}
\providecommand{\theoremname}{Theorem}

\begin{document}
\title{An Example of Ensemble Kalman Filter with Resampling}
\author{Sylvain Rubenthaler}
\curraddr{Laboratoire J. A. Dieudonn�, Parc Valrose, Universit� C�te d'Azur}
\email{rubentha@unice.fr}
\thanks{The author thanks Edoardo Calvello for providing the code from \cite{calvello-reich-stuart-2024}
via e-mail.}
\begin{abstract}
This paper introduces the \emph{Exact Ensemble Kalman Filter} (ExEnKF),
a novel algorithm for state estimation in discrete-time nonlinear
filtering problems with linear observations. Unlike traditional Ensemble
Kalman Filters (EnKFs), which approximate the filtering distribution
using ensembles of Dirac measures, the ExEnKF employs Gaussian measures,
enabling more efficient exploration of the state space and potentially
alleviating the curse of dimensionality. We prove the algorithm\textquoteright s
asymptotic consistency with the optimal filter (Theorem \ref{thm:convergence}),
establishing a convergence rate of order $1/\sqrt{N}$\LyXZeroWidthSpace{}
for $N$ particles. Numerical experiments on the Lorenz-96 multiscale
model demonstrate that the ExEnKF outperforms the standard EnKF under
model misspecification and poor initialization, particularly in highly
stochastic regimes. The algorithm\textquoteright s robustness is further
highlighted by its ability to track hidden components of the true
signal, even when observations are generated from a different model
(e.g., multiscale vs. single-scale). This work advances the theoretical
understanding of ensemble methods in nonlinear filtering and provides
a practical alternative to sequential Monte Carlo methods for high-dimensional
systems
\end{abstract}

\maketitle

\section{Introduction}

Nonlinear filtering is a cornerstone of state estimation in dynamic
systems, with applications ranging from meteorology to signal processing
and geophysical modeling. While the Kalman filter framework is optimal
for linear Gaussian systems, it struggles with nonlinearities and
non-Gaussian noise, thus requiring approximations. Ensemble Kalman
Filters (EnKFs, introduced in \cite{evensen-1994}) address this challenge
by representing the forecast/prior distribution as an ensemble of
particles and computing the analysis/posterior distribution in a manner
that mimics the Kalman filter update. Here, the terms forecast/prior
and analysis/posterior distributions stem from Bayesian analysis and
are discussed in \cite{cotter-reich-2015}. Building on \cite{calvello-reich-stuart-2024},
we refer the reader to this work for a comprehensive review of Ensemble
Kalman methods. However, the analysis of the EnKF's accuracy in with
repect to the true filtering distribution remains in its infancy.
In its regard, \cite{calvello-monmarche-stuart-2024} recently demonstrated
that in a near-linear regime, the distance between the EnKF and the
theoretical filtering distribution can be bounded. This suggests that
the EnKF may not provide a consistent approximation of the filtering
distribution. 

Alternative methods, such as Sequential Monte Carlo (SMC), can be
designed to be consistent with the underlying nonlinear filtering
problem, and do not rely on exactness only for linear Gaussian problems.
The monographs \cite{doucet-de-freitas-gordon-2001,chopin-papaspiliopoulos-2020}
provide an overview use of sequential Monte Carlo methods for general
discrete time filtering and inference problems, while \cite{del-moral-1997,del-moral-guionnet-2001}
establish their convergence of sequential Monte Carlo methods, including
in specific cases over long time horizons. However, SMC methods suffer
from the curse of dimensionality and are currently not directly applicable
to high dimensional problems such as those arising in geophysical
applications.

In the paper, we focus on a discrete time stochastic filtering problem
(i.e. a state estimation problem) where a stochastic dynamical system
is observed through linear measurements with noise. Within this relatively
restricted setting, we propose an algorithm inspired by SMC particle
methods, where Dirac masses are replaced by Gaussian measures. This
enables more efficient cover of the state space (potentially mitigating
the curse of dimensionality). Our main contributions are as follows:
we introduce an algorithm called the ExEnKF (Algorithm \ref{alg:ExEnKF}),
prove its asymptotic consistency with the optimal filter (Theorem
\ref{thm:convergence}) and demonstrate its robustness to poor initialization
and model misspecification (Figure \ref{fig:Trajectories-of-the})
in scenarios where the EnKF fails. 

The paper is organized as follows. Section \ref{sec:Presentation-of-the}
presents the model. Section \ref{sec:Mathematical-result} establishes
our convergence result. In Section \ref{sec:Numerical-simulations},
we introduce a toy model and conduct simulations for a specific set
of parameters. Under this setting, we show that the EnKF is insensitive
to poor initialization and model misspecification, and performs better
than the EnKF. 

\section{Presentation of the model\protect\label{sec:Presentation-of-the}}

\subsection{Notations}

Let $d\in\N^{*}$. For any $x\in\R^{d}$, let $|x|$ denote the Euclidean
norm of $x$. Let $\mu$ be a probability measure on $\R^{d}$ (endowed
with its Borel $\sigma$-algebra) and let $Q$ be a Markov kernel
on $\R^{d}$. Let $\phi\,:\,\R^{d}\rightarrow\R$ and $\psi\,:\,\R^{d}\rightarrow\R^{+*}$
be Borel-measurable functions. We define the following:
\begin{itemize}
\item The pushforward measure $\mu Q$ on $\R^{d}$:
\[
\mu Q(dx)=\int_{\R^{d}}\mu(dy)Q(y,dx)\,.
\]
\item The expectation of $\phi$ under $\mu$:
\[
\langle\mu,\phi\rangle=\int_{\R^{d}}\phi(x)\mu(dx)\,.
\]
\item The action of $Q$ on $\phi$:
\[
Q\phi(x)=\int_{y\in\R^{d}}Q(x,dy)\phi(y)\,.
\]
\item The supremum norm of $\psi$:
\[
\Vert\psi\Vert_{\infty}=\sup_{x\in\R^{d}}|\psi(x)|\,.
\]
\item If $\langle\mu,\psi\rangle>0$, the reweighted measure $\psi\bullet\mu$:
\[
\forall f\in\mathcal{M}_{1},\,\psi\bullet\mu(f)=\frac{\mu(f\psi)}{\langle\mu,\psi\rangle}\,,
\]
where $\mathcal{M}_{1}$ denotes the set of bounded measurable functions
on $\mathbb{R}^{d}$.
\end{itemize}
For a symmetric matrix $M\in\R^{d\times d}$ and vector $m\in\R^{d}$,
let $\mathcal{N}(m;M)$ denote the Gaussian distribution with mean
$m$ and covariance $M$. For $n\in\N^{*}$, let $I_{n}$ denote the
$n\times n$ identity matrix. Let $\mathcal{C}(\dots)$ denote the
space of continuous functions. A dot ($\dot{\dots}$) above a time
dependent function (e.g. $\dot{v}$) denotes its time-derivative.
For any $n\in\N^{*}$, we define $[n]=\{1,2,\dots,n\}$.

Following \cite{calvello-reich-stuart-2024}, \cite{calvello-monmarche-stuart-2024},
we use the Mahalanobis norm $|u|_{M}^{2}=u^{\top}M^{-1}u$ (for all
$d\in\N^{*}$, for all $u\in\R^{d}$, $M\in GL_{d}(\R)$).

\subsection{Model.}

Consider a Markov process $(X_{t})_{t\geq0}$ in $\R^{d_{X}}$ ($d_{X}\in\N^{*}$)
defined by the recurrence relations
\begin{equation}
X_{t}=\Psi(X_{t-1})+V_{t}\,,\label{eq:def-X}
\end{equation}
where 
\begin{itemize}
\item $\Psi\,:\,\R^{d_{X}}\rightarrow\R^{d_{X}}$ is a deterministic drift
function,
\item $(V_{t})_{t\geq1}$ are independent and identically distributed (i.i.d.)
centered Gaussian random variables with covariance matrix $\Sigma\in\R^{d_{X}\times d_{X}}$. 
\end{itemize}
Let $K$ denote the Markov kernel associated to $(X_{t})_{t\geq0}$.

We observe the process through linear measurements $(Y_{t})_{t\geq1}\in\R^{d_{Y}}$
($d_{Y}\in\N^{*}$), given by:
\begin{equation}
Y_{t}=H\times X_{t}+W_{t}\,,\label{eq:def-Y}
\end{equation}
where 
\begin{itemize}
\item $H$ in $\R^{d_{Y}}\times\R^{d_{X}}$ is the observation matrix,
\item $(W_{t})_{t\geq1}$ are i.i.d. Gaussian random variables with covariance
matrix $\Gamma\in\R^{d_{Y}\times d_{Y}}$ (independent of $(V_{t})_{t\geq1}$).
\end{itemize}
Our goal is to compute the filtering distribution $\pi_{t}=\mathcal{L}_{\pi_{0}}(X_{t}|Y_{1},\dots,Y_{t})$
for all $t\geq1$, where the subscript $\pi_{0}$ means that $X_{0}$
is distributed according to $\pi_{0}$. This model is standard and
be found in \cite{del-moral-1998}, example 3, p. 483.

\subsection{Algorithm}

\subsubsection{Initialization}

We sample $N$ independent particles $(X_{0}^{i})_{1\leq i\leq N}$
from the initial distribution $\pi_{0}$ and approximate $\pi_{0}$
by the empirical measure:
\[
\pi_{0}^{N}:=\frac{1}{N}\sum_{k=1}^{N}\delta_{X_{0}^{k}}\,.
\]

\subsubsection{Propagation step}

For $t\geq1$, suppose we have an approximation $\pi_{t-1}^{N}$ of
$\pi_{t-1}$ $(t\geq1)$: 
\[
\pi_{t-1}^{N}=\frac{1}{N}\sum_{k=1}^{N}\delta_{X_{t-1}^{i}}
\]
for some particles $X_{t-1}^{1},\dots,X_{t-1}^{N}$). Each particle
evolves deterministically under $\Psi$ and generates a Gaussian distribution:
\begin{equation}
\forall i\in[N]\,,\,\widehat{\pi}_{t}^{i}=\mathcal{N}(\Psi(X_{t-1}^{i}),\Sigma)\,.\label{eq:ens-gaussiennes}
\end{equation}
We call these Gaussians our \uline{forecast} ensemble. 

\subsubsection{Correction step}

We then compute a \uline{weight} for each Gaussian distribution.
\[
\mu_{t}^{i}=\int_{\R^{d_{X}}}\frac{\exp(-\frac{1}{2}(u-\Psi(X_{t-1}^{i}))^{\top}\Sigma^{-1}((u-\Psi(X_{t-1}^{i})))\exp(-\frac{1}{2}(Y_{t}-Hu)^{\top}\Gamma^{-1}(Y_{t}-Hu))}{\sqrt{(2\pi)^{d_{X}}\det(\Sigma)}\sqrt{(2\pi)^{d_{Y}}\det(\Gamma)}}du\,.
\]
 We need the following assumption.

\begin{hypothesis}$\Sigma^{-1}+H^{\top}\Gamma^{-1}H$ is invertible\footnote{In the case $\Sigma=\sigma^{2}I_{d_{X}}$ ($\sigma>0$), $H=\left[\left.I_{d_{Y}}\right|0\right]\in\R^{d_{Y}\times d_{X}}$
($d_{Y}\leq d_{X}$), $\Gamma=\mu^{2}I_{d_{Y}}$ ($\mu>0$), this
assumption holds.}.\end{hypothesis}

We have, for all $i$,
\begin{multline*}
\exp\left(-\frac{1}{2}(u-\Psi(X_{t-1}^{i}))^{\top}\Sigma^{-1}((u-\Psi(X_{t-1}^{i})))\exp(-\frac{1}{2}(Y_{t}-Hu)^{\top}\Gamma^{-1}(Y_{t}-Hu)\right)=\\
\exp\left(-\frac{1}{2}|u-\Psi(X_{t-1}^{i})|_{\Sigma}^{2}-\frac{1}{2}|Y_{t}-Hu|_{\Gamma}^{2}\right)=\\
\text{(we\,write\,\ensuremath{\Psi}\,for \ensuremath{\Psi(X_{t}^{i})})}\\
\exp\left(-\frac{1}{2}|u-(\Sigma^{-1}+H^{\top}\Gamma^{-1}H)^{-1}(\Sigma^{-1}\Psi+H^{T}\Gamma^{-1}Y_{t})|_{(\Sigma^{-1}+H^{\top}\Gamma^{-1}H)^{-1}}^{2}\right.\\
\left.-\frac{1}{2}|\Psi|_{\Sigma}^{2}-\frac{1}{2}|Y_{t}|_{\Gamma}^{2}+\frac{1}{2}|(\Sigma^{-1}\Psi+H^{T}\Gamma^{-1}Y_{t})|_{\Sigma^{-1}+H^{\top}\Gamma^{-1}H}^{2}\right)\,.
\end{multline*}
So the weight
\begin{multline}
\mu_{t}^{i}:=\frac{\exp\left(-\frac{1}{2}|\Psi|_{\Sigma}^{2}-\frac{1}{2}|Y_{t}|_{\Gamma}^{2}+\frac{1}{2}|(\Sigma^{-1}\Psi+H^{T}\Gamma^{-1}Y_{t})|_{\Sigma^{-1}+H^{\top}\Gamma^{-1}H}^{2}\right)}{\sqrt{(2\pi)^{d_{X}}\det(\Sigma)}\sqrt{(2\pi)^{d_{Y}}\det(\Gamma)}}\\
\times\sqrt{(2\pi)^{d_{X}}\det((\Sigma^{-1}+H^{\top}\Gamma^{-1}H)^{-1})}\label{eq:weight}
\end{multline}
is explicit. 

The forecast ensemble $\{\widehat{\pi}_{t}^{i}\}_{1\leq i\leq N}$
is updated into a Gaussian mixture:
\begin{equation}
\frac{\sum_{i=1}^{N}\mu_{t}^{i}\mathcal{N}((\Sigma^{-1}+H^{\top}\Gamma^{-1}H)^{-1}(\Sigma^{-1}\Psi(X_{t}^{i})+H^{T}\Gamma^{-1}Y_{t});(\Sigma^{-1}+H^{\top}\Gamma^{-1}H)^{-1})}{\sum_{i=1}^{N}\mu_{t}^{i}}\,.\label{eq:gaussian-mixture}
\end{equation}
We then sample N independant points $(X_{t}^{i})_{1\leq i\leq N}$
from this mixture. We call these points our \uline{analysis} ensemble.
From these points, we get an empirical measure we call 
\[
\pi_{t}^{N}=\frac{1}{N}\sum_{i=1}^{N}\delta_{X_{t}^{i}}\,.
\]

We refer to this method as the Exact Ensemble Kalman Filter (ExEnKF),
as it is asymptotically consistent with the optimal filter (see Theorem
\ref{thm:convergence} below). The pseudocode is provided in \ref{alg:ExEnKF}
\footnote{code available at \url{https://framagit.org/rubentha/exenkf/-/tree/81bea61c743f7d55211df037aaccc6b541515ac9/}}.
\begin{algorithm}[h]
\begin{lyxcode}
Sample~$x_{0}^{1},\dots,x_{0}^{N}$~of~law~$\pi_{0}$

For~$t$~in~$\{1,2,\dots,T\}$~do:
\begin{lyxcode}
forecast~$\widehat{x}_{t}^{i}=\Psi(x_{t-1}^{i})$~($i=1,\dots,N$)

compute~weight~$\mu_{t}^{i}$~($i=1,\dots,N$)~via~(\ref{eq:weight})

sample~$x_{t}^{1},\dots,x_{t}^{N}$~from~Gaussian~mixture(\ref{eq:gaussian-mixture})

define~$\pi_{t}^{N}=\frac{1}{N}\sum_{i=1}^{N}\delta_{x_{t}^{i}}$
\end{lyxcode}
return~$\pi_{T}^{N}=\frac{1}{N}\sum_{i=1}^{N}\delta_{x_{T}^{i}}$
\end{lyxcode}
\caption{Exact Ensemble Kalman Filter (ExEnKF)\protect\label{alg:ExEnKF}}
\end{algorithm}

\section{Mathematical result\protect\label{sec:Mathematical-result}}
\begin{thm}
\label{thm:convergence}For all $t\geq0$, there exists a constant
$C_{t}$ such that, for all $N$:
\begin{equation}
\sup_{\phi\in\mathcal{M}_{1}}\E(|\langle\pi_{t}-\pi_{t}^{N},\phi\rangle|)\leq\frac{C_{t}}{\sqrt{N}}\,,\label{eq:convergence}
\end{equation}
where $\mathcal{M}_{1}$ is the set of Borel-measurable functions
$\R^{d_{X}}\rightarrow\R$ such that $\sup_{x\in\R^{d_{X}}}|\phi(x)|\leq1$.
\end{thm}

\begin{proof}
We will prove Equation (\ref{eq:convergence}) recursively on $t$.
\begin{itemize}
\item We have (see Lemma 5.1, p. 161 of \cite{le-gland-oudjane-2004})
\[
\sup_{\phi\in\mathcal{M}_{1}}\E(|\langle\pi_{0}^{N}-\pi_{0},\phi\rangle|)\leq\frac{1}{\sqrt{N}}\,.
\]
\item Suppose we have Equation (\ref{eq:convergence}) up to time $t$.
For two probability measures $\mu$, $\mu'$ on $\R^{d_{X}}$ and
$\phi$ in $\mathcal{M}_{1}$, 
\begin{eqnarray*}
\langle\mu K-\mu'K,\phi\rangle & = & \langle\mu-\mu',K\phi\rangle
\end{eqnarray*}
and $K\phi\in\mathcal{M}_{1}$. So (as in Lemma 5.2 of \cite{oudjane-rubenthaler-2005},
or Proposition 2.5 of \cite{oudja2000})
\[
\sup_{\phi\in\mathcal{M}_{1}}\E(|\langle\widehat{\pi}_{t+1}^{N}-\pi_{t}K,\phi\rangle|)\leq\frac{C_{t}}{\sqrt{N}}\,.
\]
We set 
\[
\psi_{t+1}(u)=\exp\left(-\frac{1}{2}(Y_{t+1}-Hu)^{T}\Gamma^{-1}(Y_{t+1}-Hu)\right)\,.
\]
We have (Equation (2.7), p. 37 of \cite{oudja2000})
\begin{eqnarray*}
\sup_{\phi\in\mathcal{M}_{1}}\E(|\langle\psi_{t+1}\bullet\widehat{\pi}_{t+1}^{N}-\psi_{t+1}\bullet(\pi_{t}K),\phi\rangle|) & \leq & \frac{2\Vert\psi_{t+1}\Vert_{\infty}}{\langle\pi_{t}K,\psi_{t+1}\rangle}\sup_{\phi\in\mathcal{M}_{1}}\E(|\langle\widehat{\pi}_{t+1}^{N}-\pi_{t}K,\phi\rangle|)\\
 & \leq & \frac{2\Vert\psi_{t+1}\Vert_{\infty}}{\langle\pi_{t}K,\psi_{t+1}\rangle}\times\frac{C_{t}}{\sqrt{N}}\,.
\end{eqnarray*}
So, we get Equation (\ref{eq:convergence}) in $t+1$ with 
\[
C_{t+1}=\frac{2\Vert\psi_{t+1}\Vert_{\infty}}{\langle\pi_{t}K,\psi_{t+1}\rangle}\times C_{t}\,.
\]
\end{itemize}
\end{proof}

\section{Numerical simulations\protect\label{sec:Numerical-simulations}}

In this section, we compare the Exact Ensemble Kalman Filter (ExEnKF)
with the standard Ensemble Kalman Filter (EnKF) as described in \cite{calvello-reich-stuart-2024}
(Algorithm 2, p. 144). Our experiments are based on variants of the
Lorenz-96 model, a widely used benchmark in data assimilation and
nonlinear filtering. 

\subsection{Lorenz-96 Multiscale Model}

Let $L$, $J$ in $\N^{*}$. We have slow variables $v\in\mathcal{C}(\R^{+},\R^{L})$
and fast variables $w\in\mathcal{C}(\R^{*},\R^{L\times J})$. Each
slow variable $v_{l}$ is coupled to a group of fast variables $w_{l}=\{w_{l,j}\}_{j=1}^{J}$.
The system is governed by the following ordinary differential equations
(ODEs) (the difference between slow and fast variables is purely visual:
some oscillate faster than others). For $l=1,\dots,L$ and $j=1,\dots,J$:
\begin{equation}
\dot{v_{l}}=f_{l}(v)+h_{v}\overline{w}_{l}\,,\,\overline{w}_{l}=\frac{1}{J}\sum_{j=1}^{J}w_{l,j}\,,\label{eq:LMM-01}
\end{equation}
\begin{equation}
\dot{w}_{l,j}=\frac{1}{\epsilon}r_{j}(v_{l},w_{l})\,,\label{eq:LMM-02}
\end{equation}
where
\begin{equation}
f_{l}(v)=-v_{l-1}(v_{l-2}-v_{l+1})-v_{l}+F\,,\label{eq:LMM-03}
\end{equation}
\begin{equation}
r_{j}(v_{l},w_{l})=-w_{l,j+1}(w_{l,j+2}-w_{l,j-1})-w_{l,j}+h_{w}v_{l}\,,\label{eq:LMM-04}
\end{equation}
Boundary conditions are imposed as:
\begin{equation}
v_{l+L}=v_{l}\,,\,w_{l+L,j}=w_{l,j}\,,\,w_{l,j+J}=w_{l+1},j\,.\label{eq:LMM-05}
\end{equation}
Here:
\begin{itemize}
\item $\epsilon>0$ is a scale separation parameter, 
\item $h_{v}$, $h_{w}$ in $\R$ govern the couplings between the fast
and slow system 
\item $F>0$ is a constant forcing.
\end{itemize}

\subsection{Lorenz-96 Singlescale Model }

Let $v\in\mathcal{C}(\R^{+},\R^{L})$. For $l=1,\dots,L$, we have
ODEs
\begin{equation}
\dot{v_{l}}=f_{l}(v)+h_{v}\overline{w}_{l}\,,\,\overline{w}_{l}=M_{l}(v_{l}),\label{eq:LSM-01}
\end{equation}
where $f_{l}$ is defined in (\ref{eq:LMM-03}), and $(M_{l})_{1\leq l\leq L}$
are suitably chosen functions. 

We follow here the reasoning of \cite{calvello-reich-stuart-2024},
p. 146. If $\epsilon\ll1$, the dynamics for the $w$ governed by
(\ref{eq:LMM-02}) evolve on a much faster timescale than the dynamics
for the $v$ governed by (\ref{eq:LMM-01}). Thus it is a reasonable
approximation to think of $v$ as frozen in (\ref{eq:LMM-02}). If
we assume that the dynamics of $w$ with $v$ frozen are ergodic with
invariant measure $\mu^{v}(dw)$ (a measure in $w$, parametrized
by $v$) then the averaging principle (\cite{abdulle-weinan-engquist-vanden-eijnden-2012,vanden-eijnden-2003,pavliotis-stuart-2008})
suggests that we may make the approximation (in (\ref{eq:LMM-01}))
\[
\overline{w}\approx M(v):=\int\left(\frac{1}{J}\sum_{j=1}^{J}w_{j}\right)\mu^{v}(dw)\,.
\]
We call $M_{l}$ the $l$-th component of $M$ above. We can add the
approximation that, for all $l$, $M_{l}(v)$ is a function $m(v_{l})$
of the sole component $v_{l}$ (and not of the whole vector $v$)
(approximation that is shown to be valid for large $J$ in \cite{fatkullin-vanden-eijnden-2004}),
then we arrive at the singlescale Lorenz-96 model (\ref{eq:LSM-01}).
The function $m$ is not given explicitly, but may be estimated from
data. Figure 2.1 of \cite{calvello-reich-stuart-2024} shows such
an $m$ fit using Gaussian process regression methodology. Essentially,
you need to make a simulation of $(v,w)$ solution of (\ref{eq:LMM-01})-(\ref{eq:LMM-02})
then fit a function $m$ such that 
\[
\frac{1}{J}\sum_{j=1}^{J}w_{l,j}=m(v_{l})\,,\,\forall l\,.
\]

\subsection{Numerical experiments}

We wish to compare our algorithm to the EnKF of \cite{calvello-reich-stuart-2024}
(the code is available at \url{https://github.com/EdoardoCalvello/EnsembleKalmanMethods/}). 

\subsubsection{Generating the true signal}

To create a benchmark dataset, we simulate the true signal using the
Lorenz-96 multiscale model. We introduce a positive $\tau$ ($\tau$
will be the observation time interval). Let $L=9$, $J=8$. We initialize
the system with arbitrary values for $(v_{0},w_{0})=:x_{0}^{\dagger}\in\R^{L}\times\R^{L\times J}$
(all the components set to zero). For $t=0,1,2,\dots,T_{\max}-1$
($T_{max}\in\N$), we proceed as follows
\begin{itemize}
\item Simulate the Lorenz-96 dynamics over a time interval of length $\tau$
starting from $x_{t}^{\dagger}$ yielding a vector $\varphi(x_{t}^{\dagger})$.
\item Update the true signal
\[
x_{t+1}^{\dagger}=\varphi(x_{t}^{\dagger})+\sigma Z_{t+1}
\]
 (where $Z_{t+1}\sim\mathcal{N}(0,I_{L+J})$ is independent of $Z_{1},\dots,Z_{t}$
and $\sigma>0$. As in \cite{calvello-monmarche-stuart-2024,calvello-reich-stuart-2024},
the superscript $\dagger$ indicate that $x_{t}^{\dagger}$ is the
true signal, and remain fixed throughout the experiment.
\end{itemize}

\subsubsection{Generating the true observations}

We define the observation operator $H\in\R^{9\times6}$ as
\[
H=\left[\begin{array}{ccccccccc}
1 & 0 & 0 & 0 & 0 & 0 & 0 & 0 & 0\\
0 & 1 & 0 & 0 & 0 & 0 & 0 & 0 & 0\\
0 & 0 & 0 & 1 & 0 & 0 & 0 & 0 & 0\\
0 & 0 & 0 & 0 & 1 & 0 & 0 & 0 & 0\\
0 & 0 & 0 & 0 & 0 & 0 & 1 & 0 & 0\\
0 & 0 & 0 & 0 & 0 & 0 & 0 & 1 & 0
\end{array}\right]\,.
\]
For $(x_{1},x_{2},\dots,x_{9})^{T}\in\R^{L}$, the observation is
\[
Hx^{T}=(x_{1},x_{2},x_{4},x_{5},x_{7},x_{8})^{T}\,.
\]
The true observations $y_{t}^{\dagger}$ are generated as: 
\begin{equation}
y_{t}^{\dagger}=H\overline{x}_{t}^{\dagger}+\gamma\eta_{t}\,,\label{eq:true-observations}
\end{equation}
where
\begin{itemize}
\item $\overline{x}_{t}^{\dagger}$ is the $L$ first components of $x_{t}^{\dagger}$
(the slow variables $v$),
\item $\eta_{t}\sim\mathcal{N}(0,I_{6})$ are independent standard Gaussian
random variables,
\item $\gamma>0$ controls the observation noise.
\end{itemize}
As in \cite{calvello-monmarche-stuart-2024,calvello-reich-stuart-2024},
the superscript $\dagger$ indicates that $y_{t}^{\dagger}$ is the
true observations and remain fixed throughout the experiment. 

\subsubsection{Running ExEnKF}

We apply Algorithm \ref{alg:ExEnKF} (ExEnKF) to the observations
$(y_{t}^{\dag})_{t\geq1}$. To do so, we use the function $m$ estimated
in \cite{calvello-reich-stuart-2024}. For an initial condition $v_{0}\in\R^{L}$,
we define $\Psi(v_{0})$ as the solution of the singlescale Lorenz-96
model (\ref{eq:LSM-01}) at time $\tau$, with $M_{l}=m$ for all
$l$. We then run Algorithm \ref{alg:ExEnKF} (loop $T_{\text{max}}$
times) with:
\begin{itemize}
\item the dynamics $\Psi=\varphi$, 
\item the observations $(y_{t}^{\dagger})_{t\geq1}$,
\item a fixed number of particles $N$. 
\end{itemize}
This yields empirical measures $(\pi_{t}^{N})_{1\leq t\leq T_{\max}}$.

\subsubsection{Model misspecification and poor initialization}

To test the robustness of ExEnKF, we intentionally introduce two challenges:
\begin{enumerate}
\item Model misspecification: The observations $(y_{t}^{\dagger})_{t\geq1}$
are generated using the multiscale Lorenz-96 model, but ExEnKF assumes
the singlescale model for state estimation.
\item Poor initialization: The particles in ExEnKF are initialized with
$x_{0}^{i}\sim\mathcal{N}(10,10\times I_{L})$, rather than the true
initial condition $x_{0}^{\dagger}=0$.
\end{enumerate}
For each time $t$, we compute the empirical mean $\langle\pi_{t}^{N},\phi_{3}\rangle$
where $\phi_{3}$ is defined by $\phi(x_{1},\dots,x_{9})=x_{3}$ (the
third component of the state). This allows us to compare the estimated
trajectory of the hidden component $x_{3}$ with the true signal $\phi_{3}(x_{t}^{\dagger})$.
Note that the third component of $x_{t}^{\dagger}$ does not appear
in the observations $y_{t}^{\dagger}$ (see $H$ above), making this
a stringent test of the algorithm's ability to track unobserved components. 

\subsubsection{Results}

\begin{figure}
\begin{centering}
\begin{minipage}[t]{0.48\columnwidth}%
\begin{center}
\includegraphics[scale=0.5]{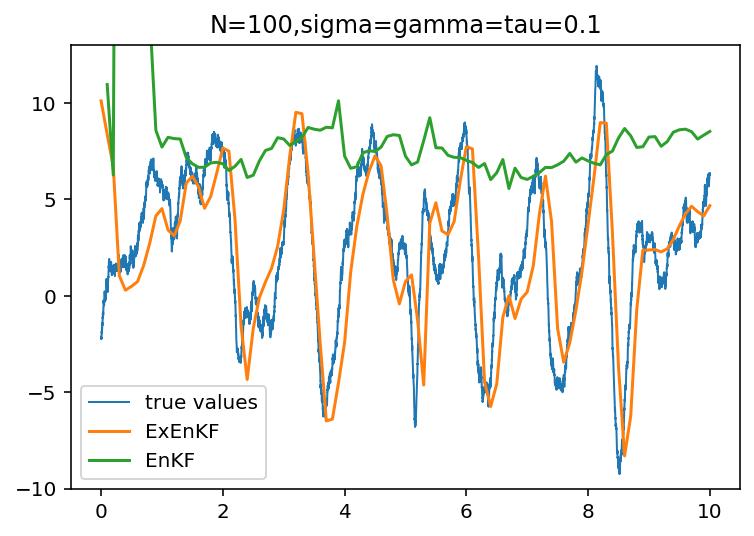}
\par\end{center}%
\end{minipage}\hphantom{}%
\begin{minipage}[t]{0.48\columnwidth}%
\begin{center}
\includegraphics[scale=0.5]{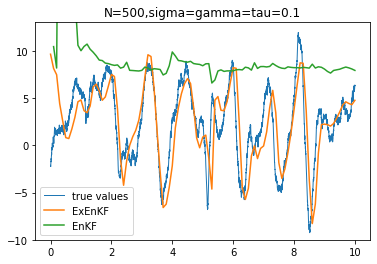}
\par\end{center}%
\end{minipage}\caption{\protect\label{fig:Trajectories-of-the}Trajectories of the third
component of the signal (true and estimated).}
\par\end{centering}
\end{figure}
We present the results in Figure \ref{fig:Trajectories-of-the} which
shows the true trajectory (blue) alongside the estimates from EnKF
(green) and ExEnKF (orange). In all experiments, we set 
\[
\sigma=\gamma=\tau=0.1.
\]
These parameters correspond to a highly stochastic regime, where traditional
EnKF methods often struggle. As shown in Figure 2.5 of \cite{calvello-monmarche-stuart-2024},
EnKF performs well for $\sigma=\gamma=0.1$ and $\tau=10^{-3}$. EnKF's
principle is that you run computations similar to those used for a
Kalman filter (where all the measures are Gaussian and $\Psi$ is
a linear operator). When $\tau$ increases, this approximation is
bound to fail (see Figure \ref{fig:Trajectories-of-the}). We also
see in Figure \ref{fig:Trajectories-of-the} that ExEnKF manages to
track down the true trajectory after a poor initialization at time
$0$. We tried the various algorithms with $N=100$ and $N=500$ without
noticeable variation.

Like in \cite{calvello-monmarche-stuart-2024}, the algorithm struggles
to recover the original trajectory when $\tau$ increases (see Figure
\ref{fig:Trajectories-of-the-1}).
\begin{figure}
\begin{minipage}[t]{0.45\columnwidth}%
\includegraphics[scale=0.5]{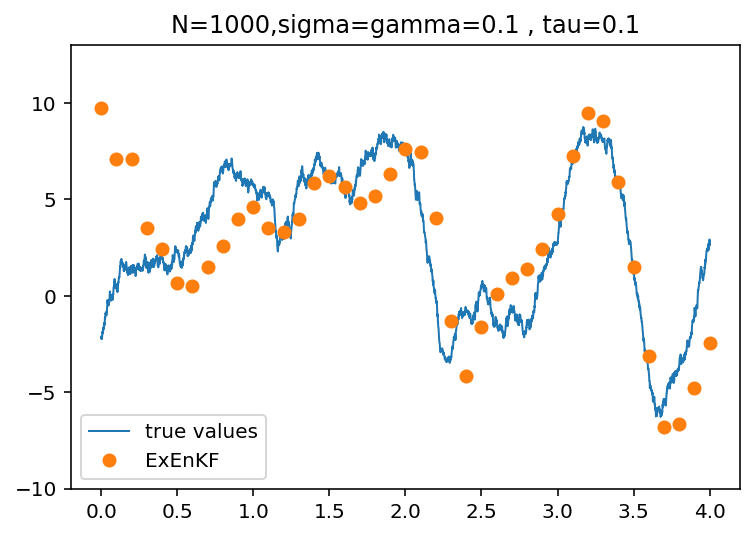}%
\end{minipage}\hphantom{}%
\begin{minipage}[t]{0.45\columnwidth}%
\includegraphics[scale=0.5]{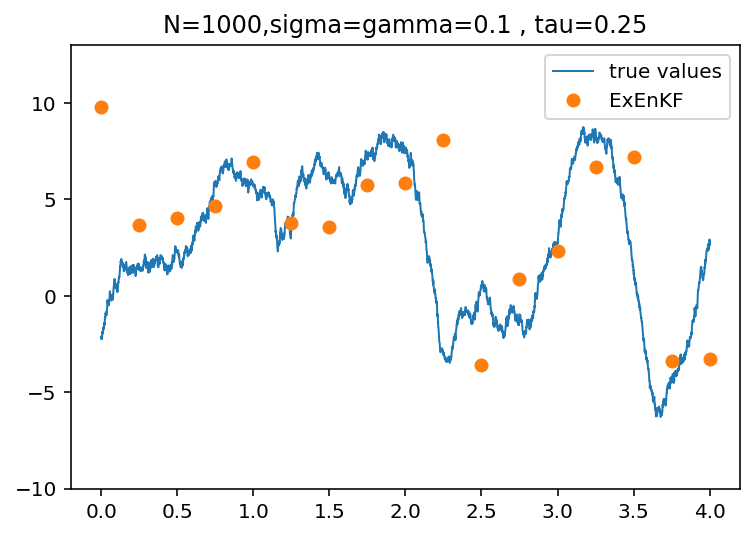}%
\end{minipage}

\caption{Trajectories of the third component of the signal (true and estimated
by ExEnKF) for\protect\label{fig:Trajectories-of-the-1} large $\tau$.}

\end{figure}

\begin{quote}
\bibliographystyle{amsalpha}
\bibliography{EnKF}
\end{quote}

\end{document}